\documentclass{article}
\usepackage[margin=1in]{geometry}
\usepackage{times} 

\usepackage{mathtools, amsfonts, amsthm, thmtools, thm-restate}
\usepackage{algorithm}
\usepackage{algorithmic}
\usepackage{enumitem}
\usepackage{subcaption}
\usepackage{booktabs}
\usepackage{hyperref}
\usepackage{cleveref}
\usepackage{url}

\usepackage[sortcites=true]{biblatex}
\newcommand \reals {\mathbb{R}}
\newcommand \prob {\operatorname{Pr}}

\newcommand \qsingle {q_\text{single}}
\newcommand \qwithin {q_\text{within}}
\newcommand \qacross {q_\text{across}}

\newcommand \topics {{\mathbb O}}
\newcommand \users {{\mathbb U}}

\newtheorem{thm}{Theorem}

\newtheorem{definition}[thm]{Definition}
\newtheorem{proposition}[thm]{Proposition}

\title{Differentially Private Synthetic Data Release for Topics API Outputs}

\author{%
Travis Dick\\
  Google Research\\
  \texttt{tdick@google.com} \\
  \and
Alessandro Epasto\\
  Google Research\\
\texttt{aepasto@google.com}\\
\and
Adel Javanmard\\
 USC, Google Research\\
\texttt{ajavanma@usc.edu}\\
\and
Josh Karlin\\
Google Chrome\\
\texttt{jkarlin@google.com}\\
\and 
Andr\'es Mu\~noz Medina\\
Google Chrome\\
\texttt{ammedina@google.com}\\
\and
Vahab Mirrokni\\
Google Research\\
\texttt{mirrokni@google.com}\\
\and
Sergei Vassilvitskii\\
Google Research\\
\texttt{sergeiv@google.com}\\
\and 
Peilin Zhong\\
Google Research\\
\texttt{peilinz@google.com}\\
}

\begin{document}

\maketitle

\begin{abstract}
The analysis of the privacy properties of Privacy-Preserving Ads APIs is an area of research that has received strong interest from academics, industry, and regulators. Despite this interest, the empirical study of these methods is severely hindered by the lack of publicly available data. Reliable empirical analysis of the privacy properties of an API, in fact, requires access to a dataset consisting of realistic API outputs for a large collection of users; however, privacy concerns prevent the general release of such data to the public.

In this work, we address this problem by developing a novel methodology to construct synthetic API outputs that are simultaneously realistic enough to enable accurate study and provide strong privacy protections. We focus on one of the Privacy-Preserving Ads APIs: the Topics API, part of Google Chrome's Privacy Sandbox, which enables interest-based advertising without relying on third-party cookies. We developed a methodology to generate a differentially-private dataset that closely matches the re-identification risk properties of the real Topics API data. The use of differential privacy provides strong theoretical bounds on the leakage of private user information from this release.

Our methodology is based on first computing a large number of differentially-private statistics describing how output API traces evolve over time. Then, we design a parameterized distribution over sequences of API traces and optimize its parameters so that they closely match the statistics obtained. Finally, we create the synthetic data by drawing from this distribution.

Our work is complemented by an open-source release of the anonymized dataset obtained by this methodology. We hope this will enable external researchers to analyze the API in-depth and replicate prior and future work on a realistic large-scale dataset. We believe that this work will contribute to fostering transparency regarding the privacy properties of Privacy-Preserving Ads APIs.
\end{abstract}

\section{Introduction}
Advertisement personalized to a user's browsing behavior is a fixture of the modern internet. The ability to target ads based on a user's latent interests has many benefits for the ecosystem, as personalized ads are often more relevant to the user, obtain better performance for the advertiser, and achieve higher revenue for the publishers. All these benefits, which fund the free web, however, come with a privacy cost. Traditionally, to sustain personalized advertisements, many adtechs have extensively used third-party cookies to amass large and detailed profiles of users' browsing activities \cite{englehardt2016online}. As the extensive reliance on third-party cookies incurs privacy risks, many browser vendors have recently proposed a variety of Privacy-Preserving Ads (PPA) APIs to enable personalized ads without the use of third-party cookies. These PPA APIs  include, for instance, the Chrome Privacy Sandbox's APIs~\cite{privacysandbox}, Microsoft Edge's Private Ad Selection API \cite{microsoftprivatead} and Apple's AdAttributionKit \cite{appleadattribution}. The main goal of all these initiatives is to significantly reduce the privacy risk of advertisement compared to methods using third-party cookies.

Given the widespread deployment on the web of such PPA APIs, academics, browser vendors, and regulators alike have all shown interest in analyzing the privacy properties of such methods~\cite{MozillaFloc,topics-explainer,carey2023measuring,jha2024re,alvim2024privacy,alvim2023quantitative,verna2024first,beugin2024publicreproducibleassessmenttopics,floc}. Unfortunately, however, the empirical study of these proposed methods by academic researchers is hindered by the lack of publicly available data on the output of such APIs. 

Researchers not affiliated with the internet industry, for obvious reasons including privacy, do not have easy access to high-quality real browsing data from a large number of users. For this reason prior work from academic researchers has often relied on small scale data (few thousands of users) or on fully synthetic data (not based on real browsing histories)~\cite{beugin2024publicreproducibleassessmenttopics,jha2024re,alvim2024privacy,alvim2023quantitative}. Using data that does not match real browsing histories of a representative large population can skew privacy analysis arbitrarily and in unpredictable ways. Moreover, the accurate replication of prior work from industry researchers with access to large-scale real data~\cite{topics-explainer,carey2023measuring} is not feasible on a small-scale dataset. This hinders the ability to credibly verify (or disprove) the claims of browsing vendors. 

To address this problem, in this work, we design a methodology to release a privacy-preserving large-scale dataset of PPA APIs output traces. We focus on one API, Chrome's Privacy Sandbox Topics API~\cite{topics-github} which has been extensively studied in the past in terms of its re-identification risk~\cite{topics-explainer,carey2023measuring,jha2024re,alvim2024privacy,alvim2023quantitative,verna2024first,beugin2024publicreproducibleassessmenttopics}. Our methodology, and the dataset we release, has the aim of allowing any researcher to study the re-identification risk of the API and to replicate prior work. \emph{Our data release (and links to the supporting code) are made available through kaggle~\cite{datarelease}.}

The main technical contribution of our work is a methodology to construct synthetic API outputs that mimic the real API traces while providing strong provable privacy guarantees. To ensure privacy, we use the gold-standard provable privacy protection of Differential Privacy (DP), mitigating the risk that our release leaks information about individual users. Thanks to DP, our methodology allows us to safely provide, for the first time, a release based on the real browsing data of hundreds of millions of users (as opposed to prior work with thousands of users~\cite{beugin2024publicreproducibleassessmenttopics,jha2024re}).

In terms of accuracy we show that our methodology generates traces that have similar re-identification risk to the real data, thus facilitating external researchers in verifying the claims in prior work and in studying the Topics API in depth. We believe that large scale data releases like ours will further deepen the understanding of the research community on the properties of Privacy Preserving Ads technologies.

\subsection{Methodology}

Our methodology to obtain a privacy-preserving data release is based on four steps.

{\bf Statistics extraction.} The first step analyzes the Topics API data of a large number of users---see Section~\ref{sec:topics-api} for details on how the API works. From this data, we define a large number of statistics describing real users' topic traces and their evolution over time. The statistics include the point-wise distribution of the topics in output and their pairwise co-occurrence within (and across) periods of time. Crucially, we obtain these aggregate statistics with strong privacy protection using differential privacy--see Section~\ref{sec:stat} for details on how the statistics are obtained. The rest of the data release methodology uses as input {\it exclusively} these anonymous, highly aggregated, differentially private statistics, thus bounding the privacy risk of all downstream modeling steps.

{\bf Model optimization.} Next, we design a parameterized model describing a distribution over sequences of Topics API data. The parameters of the model are optimized with stochastic gradient descent methods to reduce the discrepancy between the statistics induced by the model on the data and the statistics obtained in the previous step---see Section~\ref{sec:model} for how the model is defined and trained. This model is, by the post-processing properties of differential privacy, itself differentially private.

{\bf Data sampling from the model.} Finally, we use the prior obtained DP model to create the synthetic Topics API data by drawing i.i.d. samples from its distribution (see \Cref{alg:sampling}).

Notice that both the model and data sampled from it are derived solely by post-processing the DP statistics obtained in the first step. Hence, we can publish the statistics, the parameters of the model, and the data sampled while satisfying differential privacy.

{\bf Verification of the data accuracy.} After we generate the data, we verify its quality by empirically comparing its properties with the original data. First, we verify that the synthesized data very closely matches the statistics observed in the real data.

Then, we follow the methodology of prior work~\cite{carey2023measuring} and empirically verify that the synthetic data has good agreement with the re-identification rate of the real data (see~\cref{sec:exp}). This allows us to replicate prior work with good accuracy.

\section{Related work}
\label{sec:rel}
Understanding the effect of different web technologies on user privacy has been a major focus of the privacy and security community for decades. In this area of work, the privacy risks involved in the use of third-party cookies are well-understood~\cite{6234427}.

Another well-established area of research is the measurement and mitigation of the so-called browser fingerprinting risk~\cite{browser_fingerprinting,unique_browser,ml_for_fingerprinting,fingerprintingfr,munir2023cookiegraph}, i.e., the potential for tracking users based on information provided in HTTP request headers or extracted from other web APIs (e.g., user agent, screen size, fonts, etc.).

As web browsers have reduced access to third-party cookies while increasingly deploying~\cite{johnson2024unearthing} a variety of Privacy Preserving Ads APIs, the privacy community has devoted significant attention to understanding the privacy implications of these APIs.

Some research in this space has focused on bounding the theoretical worst-case privacy risk properties of APIs, for instance through the lens of differential privacy, as in the case of the Attribution Reporting API~\cite{ghazi2024differential}.

Another research direction, which is the focus of this work, has been understanding the impact of these APIs in the wild by empirically analyzing real browsing data. In this body of work, the most relevant research to our paper is the study of the re-identification risks of the Topics API~\cite{topics-explainer,carey2023measuring,jha2024re,alvim2024privacy,alvim2023quantitative,verna2024first,beugin2024publicreproducibleassessmenttopics}. These papers have focused on assessing the likelihood of an attacker correctly matching the identity of a user across two sites based on the Topics API outputs observed on both sites~\cite{carey2023measuring}. Results in these analyses have varied widely (even using similar methodology) due to the vastly different underlying datasets used to measure the risk.

For instance, consider the size of the data, which is known to affect the re-identification risk~\cite{carey2023measuring}, as re-identifying a user is simpler among fewer other users. While some industry researchers~\cite{carey2023measuring,topics-explainer} could perform their analysis on proprietary datasets containing hundreds of millions of real users from different countries, several academic studies have instead been based on datasets with around $100-2000$ users~\cite{beugin2024publicreproducibleassessmenttopics,jha2024re}. The work of Beugin et al.~\cite{beugin2024publicreproducibleassessmenttopics}, for instance, uses data from about 2000 volunteers in Germany. Other work used imputed API outputs from data collected long before the introduction of the PPA APIs (e.g., AOL logs from 2006~\cite{alvim2024privacy}) or pure theoretical analysis~\cite{alvim2023quantitative}. This is expected, as academic researchers unfortunately do not have easy access to good-quality, large-scale browsing data.

Unsurprisingly, the use of datasets with different scales, creation methodologies, and potential biases has resulted in widely different estimates of re-id risk, sometimes disagreeing by orders of magnitude.

The goal of our work is to publish an open-source dataset that is based on the real, large-scale proprietary data usually available only to tech companies, thus enabling the replication of prior work by all researchers. This dataset will also allow for comparing different study methodologies on the same data.

Finally, a related area of research is that of DP Machine Learning, and especially synthetic data generation with DP (we refer to a research survey for more details~\cite{ponomareva2023dp}). A large body of work in this area has focused on designing general-purpose DP models for generating text data~\cite{tang2023privacy}, images~\cite{lin2023differentially}, and tabular data~\cite{swanberg2025apiaccessllmsuseful}, for instance. Our work is in the line of research of synthetic data generation with DP using marginal-based methods~\cite{RAP2021}. Our paper specifically designs a custom synthetic data model for the Topics API, based on the unique properties of the API discussed in the next section, and that requires only aggregated access to statistical information on the private data.

\section{The Topics API}
\label{sec:topics-api}

\begin{algorithm}
\caption{Topics API. 
\newline {\bf Input:} Topics $\topics$, probability to return a random
topic $p$. }
\label{alg:topics}
\begin{algorithmic}
\STATE {\bf On device (at the end week $i$):} Select the set $S_i$ of top $5$ most frequent topics in the {\it on-device} websites visited by the user, and that have called the API, during the week $i$. If fewer than $5$ topics are present, pad with uniformly at random (u.a.r.) topics from $\topics$.
\STATE \textbf{On call GetTopic()} from website $w$, during week $i$, for user $u$:
\STATE Seed the random number generator with $w, i, u$.
\STATE Flip coin with heads probability $p$.
\IF{Heads}
\RETURN Element of $\topics$ chosen u.a.r. 
\ELSE
\RETURN Element of $S_{i-1}$ chosen u.a.r. 
\ENDIF
\end{algorithmic}
\end{algorithm}

As our work focuses on studying in depth the privacy properties of the Topics API~\cite{topics-github} part of Google Chrome's Privacy Sandbox~\cite{privacysandbox}, we first describe how the API works. (We refer to~\cite{topics-github} for detailed specifications of the API.)
The purpose of the Topics API is to enable interest-based advertising without the use of third-party cookies in order to reduce the risk of cross-site tracking. 
To do so, the API introduces several privacy protections including limiting the interest analyzed to a coarse-grained curated list of interests and introducing noise.  

The API has two main processes (See Algorithm~\ref{alg:topics}). First, when a client visits a website that calls the API, that website is classified (on-device) into a taxonomy $\topics$ of topics, where $|\topics|=469$ in the v2 taxonomy\footnote{\url{https://raw.githubusercontent.com/patcg-individual-drafts/topics/refs/heads/main/taxonomy_v2.md}}. This allows the API to build a weekly {\it on-device} interest profile of the user. This computation is done at the end of week $i$, when the browser builds an interest profile $S_i$ for the user that will be used the next week. This profile contains the top five topics (from the taxonomy) of the websites visited by the user during the week $i$ that have called the API. $S_i$ is guaranteed to be of size $5$ by padding it, if necessary, with random topics from $\topics$ whenever there are fewer than 5 topics visited during a week. Importantly, this profile is kept on the browser and not shared with others, and is kept fixed until the update at the end of the subsequent week. 

Then, during week $i$, whenever a website calls the API, the browser selects a topic at random from the top topics in the profile $S_{i-1}$ of the user obtained in the previous week. This correct topic is released with probability $1-p$, while a uniformly at random topic among all $\topics$ topics is selected with  probability $p$ ($p=5\%$ in the current specifications). Notice that, for every user-website pair, the topic sampled for that website for the user is fixed the entire week (i.e., two visits of the user to the website will result in the same topic). Samples from the same user (within a week) on two sites are instead drawn independently. After the week is complete, the profile is updated again as described above.\footnote{We follow the methodology of Carey et al.~\cite{carey2023measuring} and omit two complicating details from the actual API. First, the API also returns the cached results  of the previous two weeks. Modeling the real output observed by an adtech on a site for $r$ weeks simply corresponds to $r+2$ weeks of observations in our model, so we omit this detail. Second, the API applies a filtering to the output to remove topics that were not observed by an adtech on any site. We make the more pessimistic assumption that no filtering ever happens, as this provides more information to an attacker of the API.}

\section{Statistics for trace modeling}\label{sec:stat}
Our goal is to develop a model of the Topics API output observed by an adtech (or third party present on a site such as adtechs) for a set of users, over a period of time. Consider the observation of the output of the API for user $u$ on website $w$, for $r$ weeks. We will number those weeks without loss of generality $1, 2, \ldots, r$. We call $O(u,w)=(o_1(u,w), \ldots, o_r(u,w))$  the {\it trace} of the output of the API for user $u$ on website $w$ for the weeks $1$ to $r$, (i.e., $o_i(u,w)$ is the output for week $i$ for the user). Notice that by the property of the API there is a fixed $o_i(u,w)$ observed by the site $w$ (and all third-party present on the site) during week $i$ for the user even if the user visits the site multiple times. When clear from context we will omit $u, w$ and call $O = (o_1, \ldots, o_r)$ the trace of the user. Notice that $O$ describes completely the Topic API outputs observed by the website for the user during the period of time. Our goal is to learn a privacy-preserving model of the distribution $D(u)$ over traces so that we can sample $O(u,w) \sim D(u)$ for users $u \in \users$. Notice, in fact, that for any set of sites $W$, the traces $O(u,w)$, for $w\in W$, are all sampled i.i.d. from the distribution $D(u)$.  So we will focus the rest of the discussion on how to model $D(u)$.

From the definition of the API, the distribution $D(u)$ is entirely described by the sequence $S(u) = (S_0(u), \ldots, S_{r-1}(u))$ of the $5$ top topics of the user $u$ for the weeks $0$ to $r-1$. In fact, we have $$\prob \{ o_i(u,w) = j\} = (1-p) /5 + p/|\topics|$$ if $j\in S_{i-1}(u)$ and otherwise
$$\prob \{ o_i(u,w) = j\} = p/|\topics|$$
So it is sufficient for us to only model the distribution of $S(u)=(S_1(u),\ldots, S_r(u))$, for a given user, in order to model the output of the API on arbitrary many sites. This will be the focus of all the modeling efforts of the paper. 

Unfortunately, exactly modeling this distribution $S(u)$ is infeasible for even small $r$, and even neglecting the privacy constraints. This is because the support of this distribution is the set of all sequences of length $r$ of $5$ topics out of $|\topics| = 469$ topics. This scales with  ${|\topics| \choose 5}^r$ which is too large to measure accurately, or just output, for even very small $r$ (for instance for $r=1$ it contains more than 100 billion parameters!) . Since our goal is to model the long term evolution of the traces of a user for several weeks we need to accept a less general model. 

For this reason, our approach is to impose a series of simplifying assumptions on the distribution that we can model allowing a computationally efficient and privacy-preserving model design. As we will see next, our approach will be to impose that the distribution of the data we can generate must approximately match, across a large number of statistics, the same statistics of the real distribution. Our ML model will learn a set of parameters to maximize the match over those statistics. The choices of the statistics, described in the next section, is inspired by identifying a small set of the salient parameters that fully characterize the evolution of the top topics of a user under assumptions of time-stationarity of the trace evolution.   

\subsection{Statistics definition}

We now define and justify the statistics we measure from real user data and that we use to fit a parametric distribution. As we discussed above, we model the top user topic sets $S_0, \ldots, S_{r-1}$ in weeks $0,\ldots, r-1$ as these uniquely determine the distribution over topic traces that would be observed by any website for that user in weeks $1,\ldots, r$.
In our model, we define the following statistics:
\begin{itemize}[leftmargin=*]
    \item For each topic $o \in \topics$, let
    \[
    \qsingle^*(o) = \prob_{U, J}\left (o \in S_{J}(U)\right),
    \]
    where $U$ is a randomly chosen user and $J$ is a randomly chosen week in $\{0, \ldots, r-1\}$.
    There are 469 such statistics.
    \item For each pair of distinct topics $o_1, o_2 \in \topics$, \[
    \qwithin^*(o_1, o_2) = \prob_{U, J}(\{o_1, o_2\} \subset S_{J}(U)),
    \]
    where $U$ is a randomly chosen user and $J$ is a randomly chosen week in $\{0, \ldots, r-1\}$.
    There are ${469 \choose 2} = 109,746$ such statistics.
    \item For each pair of (possibly equal) topics $o_1, o_2 \in \topics$,
    \[
    \qacross^*(o_1, o_2) = \prob_{U, J}\left (o_1 \in S_{J-1}(U) \wedge o_2 \in S_{J}(U)\right),
    \]
    where $U$ is a randomly chosen user and $J$ is a randomly chosen week in $\{1, \ldots, r-1\}$.
    There are $469^2 = 219,961$ such statistics.
\end{itemize}

Roughly speaking, $\qsingle^*$ and $\qwithin^*$ measure the popularity of topics and pairs of topics, while $\qacross^*$ measures how topics change over time.
Also, notice that, since the sets $S_j(u)$ are guaranteed to have $k$ topics, we can derive $\qsingle^*$ from $\qwithin^*$:
\[
 \qsingle^*(o_1) =  \sum_{o_2\neq o_1} \qwithin^*(o_1, o_2) / (k-1).
\]

Note that all three kinds of statistics are averaged over the training data weeks and are imposed on the model for all weeks simulated. This constitutes a stationarity assumption: the topic popularity and transition dynamics do not significantly change over the data weeks. In particular, when generating data for $r'$ weeks, we enforce the popularity statistics $\qsingle^*$ and $\qwithin*$ on every week, and the transition statistics $\qacross^*$ on every pair of consecutive weeks. The main advantage is that this allows us to generate synthetic topics data for many weeks without the need to obtain a collection of statistics that grows with the number of weeks. This is computationally efficient and enables better statistical accuracy at parity of privacy guarantee as it reduces the number of parameters that must be estimated privately.

\paragraph{Summary}
The simplifying assumptions described in this section greatly decrease the parameters needed to describe the distribution of the Topics API traces. The key enabling this simplification is a stationarity assumption postulating that the probability distribution of topics and their transition probabilities are stable over time. We test this assumption in Appendix~\ref{app:stationarity} where we show that all the statistics we compute on our training data have $\ge 99.8\%$ correlation coefficient with the same statistics computed on a different period of time from the validation data, confirming that the stationarity assumption holds with high accuracy.

While an arbitrary trace distribution over $r$ weeks would require ${|\topics| \choose 5}^r$ parameters to describe, thanks to this simplifying assumption our model require estimating only  $|\topics|^2 + {|\topics| \choose 2} \approx 329{,}000$ statistics $\qwithin^*, \qacross^*$ for arbitrary long traces. The models we fit to these statistics (described in Section~\ref{sec:model}) have approximately 4.7M parameters. 
As we discuss in the next section, these statistics can be estimated efficiently and accurately in a privacy-preserving way. 
Moreover, as we observe in our empirical analysis (see Section~\ref{sec:exp}) our modeling closely matches some important aspects of the real data behavior.
Before describing in Section~\ref{sec:model} how we learn a model to sample arbitrary long traces approximately matching the observed statistics we described how to estimate $\qwithin^*, \qacross^*$ in a privacy preserving way.

\subsection{Data collected and privacy-preserving data handling}
We now describe the data used to obtain the DP statistics. Our work is based on a Google proprietary log of de-identified Topics API outputs for Chrome's users. From this dataset we obtain for more than a hundred million users the top $k=5$ topics assigned to the user for a period of time. More precisely, we based our analysis on the latest $4$ weeks of data\footnote{In line with our data access and retention policy, user ids are removed from this data and access is restricted to the latest 28 days of logs. The data was collected for the $4$ weeks between 2024-12-16 and 2025-01-12.} to obtain the  sequences $S(u)=(S_0(u), S_1(u), S_2(u), S_3(u))$ associated to a set of users $u\in \users$ of users for $r=4$ weeks. The statistics used to build the model are based on analyzing a dataset $D$ containing the first two weeks of data $D=(S_0(u), S_1(u))_{u \in \users}$ while the full $4$ weeks are used in the re-identification risk measurement (see~\Cref{sec:exp}).

\paragraph{Differentially privacy}
We briefly introduce the framework of Differential Privacy which enables us to provide strong privacy protections for the statistics computed on this data.

Differential privacy (DP)~\cite{dwork2006calibrating} (refer to~\cite{dwork2014book} for an in-depth treatment of the area) has emerged as the gold standard for ensuring strong privacy protection, when  dealing with private data. Intuitively, differential privacy promises that the output of an algorithm remains approximately the same (in distribution), whether or not a particular individual users' data is included in the input of the algorithm. This guarantee provably limits the ability of an  adversary to accurately infer any sensitive information about a specific individual (including even whether the user is present in the data or not). More formally, our differential privacy protection can be stated in the following way.  
\begin{definition}[Neighboring Datasets]
We say that datasets $D,D' \in \mathcal{D}$ containing user's topics sequences are neighboring $(D \sim D')$, if one can be obtained by adding a single user's topics sequence to the other, i.e., 
$$D =\{(S_0(u), \ldots S_r(u)) \;|\; u \in \users\}\; \mathrm{and}$$
$$D' = D \cup \{(S_0(v), \ldots S_r(v))\}$$ for a user $v\notin \users$.
\end{definition}

\begin{definition}[Differential Privacy~\cite{dwork2014book}]
A randomized algorithm $\mathcal{M}$ is $(\epsilon, \delta)$-differentially private, or $(\epsilon, \delta)$-DP, if for any two neighboring datasets $D$ and $D'$ and for any possible (measurable) subset of outputs $\mathcal{O}$ of the algorithm,
\[
    \Pr(\mathcal{M}(D) \in \mathcal{O}) \leq \exp(\epsilon) \cdot \Pr(\mathcal{M}(D') \in \mathcal{O}) + \delta.
\]
\end{definition}

To achieve DP for the statistics computed we use the well-known Gaussian mechanism. Suppose the statistics computed by the algorithm are represented as a vector-valued function $f : \mathcal{D} \to \mathbb{R}^d$ where $d$ is the number of statistics computed. The sensitivity of this function is defined as 

\begin{definition}[$\ell_2$ sensitivity]
The $\ell_2$ sensitivity $\Delta_2)$ of the function $f$ is defined as the maximum of adjacent datasets $D,D'\in \mathcal{D}$ 
\[
    \Delta_2(f) := \max_{D \sim D'} ||f(D)-f(D')||_2 
\]
\end{definition}

\begin{proposition}[Gaussian Mechanism \cite{balle2018gaussianmechanism}]\label{prop:gaussian}
Let $f: \mathcal{D} \to \mathbb{R}^d$ be a function with $\ell_2$ sensitivity $\Delta_2(f)$. For any $\epsilon > 0$ and $\delta \in (0,1]$, the mechanism $M(x) = f(x) + Z$ with $Z \sim \mathcal{N}(0, \sigma^2 I_d)$ is $(\epsilon, \delta)$-DP if
\begin{equation*}
    \Phi\left(\frac{\Delta_2(f)}{2\sigma} - \frac{\epsilon \sigma}{\Delta_2(f)}\right) - \exp(\epsilon) \Phi\left(-\frac{\Delta_2(f)}{2\sigma} - \frac{\epsilon \sigma}{\Delta_2(f)}\right) \leq \delta.
\end{equation*}
where $\Phi: \mathbb{R} \to \mathbb{R}$ be the standard Gaussian cumulative density function, and $I_d$ is the identity matrix of size $d$.
\end{proposition}

\Cref{prop:gaussian} allows setting a $\sigma$ parameter depending on the sensitivity and desired $(\epsilon, \delta)$-DP guarantees for the Gaussian noise to be added to the statistics.

For our statistics computation we used a dataset $D$ containing the first two weeks of data, over which we define the following three functions: 
\begin{enumerate}
\item $f_{11}(D)[o_i,o_j]$, for $o_i, o_j \in {\topics \choose 2}$, is the number of users with both topic $o_i$ and $o_j$ in the first week;
\item $f_{22}(D)[o_i,o_j]$ is the number of users with both topics in week two;
\item $f_{12}(D)[o_i,o_j]$, for $o_i, o_j \in {\topics ^2}$, is the number of users with topic $o_i$ in week 1 and topic $o_j$ in week 2.  
\end{enumerate}

It is easy to prove the following result: 
\begin{restatable}{lemma}{lemSensitivity}\label{lemma:sensitivity}
The $\ell_2$ sensitivity of the functions defined above is $\Delta_2(f_{11}) = \Delta_2(f_{22}) = \sqrt{{k \choose 2}}$ and $\Delta_2(f_{12}) = k$
\end{restatable}
The proof is available in the appendix.

\paragraph{Implementation details}
In our data release we computed the statistics defined in the vector $f$ above using a Google DP pipeline implementing the Gaussian mechanism~\cite{amin2022plume}
We used a total privacy budget of $\epsilon=\log(3)$, $\delta= 10^{-15}$ in line with strict privacy protection standards. Internally the count vectors in $f_{11}$ and $f_{22}$ received $25\%$ of this budget while the counts $f_{22}$ received $50\%$. Notice that from simple post-processing of the DP counts thus obtained (by normalization) we can obtain the the statistics $\qsingle^*, \qwithin^*, \qacross^*$. Specifically, we estimate $ \qwithin^*$ averaging (and normalizing) the count vectors $f_{11}, f_{22}$ while $ \qacross^*$ comes from normalizing the counts $f_{12}$ (and $\qsingle^*$ is induced by $ \qwithin^*$ as described above).

\paragraph{Description of the statistics obtained}
We now briefly describe the DP statistics $\qsingle^*, \qwithin^*, \qwithin^*$ obtained. The full DP statistics tables are provided as part of our data release. Due to space limitations, we only present some basic statistics. In \Cref{fig:stats-top}, we visualize the top topics for the $\qsingle^*$. In  Appendix~\ref{app:additional-plots} we provide the top entries in $\qwithin^*$ and $\qacross^*$. It is possible to observe that \textsc{/News}, \textsc{/Arts \& Entertainment}, and \textsc{/Shopping} are the top topics observed. For the top pairs of topics in the $\qwithin^*$ and $\qacross^*$ we observe that the distribution comes mostly from those same top topics. More details are provided in \Cref{app:additional-plots}.

\begin{figure}[ht]
\centering
    \begin{tabular}{lr}
    \toprule
    Name & Rate \\
    \midrule
    /News & 34.60\% \\
    /Arts \& Entertainment & 24.52\% \\
    /Shopping & 23.34\% \\
    /Sports & 18.31\% \\
    /Games/Computer \& Video Games & 17.70\% \\
    /Internet \& Telecom & 12.97\% \\
    /Computers \& Electronics/Software & 11.82\% \\
    /Computers \& Electronics & 10.46\% \\
    /Food \& Drink/Cooking \& Recipes & 9.72\% \\
    /Arts \& Entertainment/TV \& Video & 8.02\% \\
    \bottomrule
    \end{tabular}
    \caption{Top topics distribution as represented by the top DP statistics in $\qsingle^*$.\label{fig:stats-top}}
\end{figure}

\section{Modeling traces}\label{sec:model}

In this section, we design a model for sequences of top topic sets that is consistent with the statistics collected in \Cref{sec:stat}.
Our strategy is inspired by the Relaxed Adaptive Projection mechanism for private synthetic data~\cite{RAP2021}, and by the dataset reconstruction attack methodology of Dick et al~\cite{Census2023}.
The high level idea is as follows: we design a parameterized distribution over topic set sequences such that all of the statistics from \Cref{sec:stat} are differentiable functions of the model parameters.
Then we use gradient-based optimization techniques to search for distribution parameters such that the statistics for the parameterized distribution closely match those privately collected from real user topic set sequences.
Given the fit distribution, we generate synthetic topic set sequences by sampling users i.i.d. from it.
Due to the post-processing guarantees of Differential Privacy, the fit distribution and sampled synthetic user topic traces can be released while satisfying differential privacy.

We begin by describing the parametric form of our topic set sequence distribution.
The most basic component is a \emph{slot}, which is a distribution over a single topic.
Each slot's topic distribution is parameterized by one weight/logit for each topic, which is converted to a topic distribution by the softmax function.
In other words, each slot is a distribution over topics encoded by a logit vector.
Next, a \emph{type} represents a simple distribution over topic set sequences.
If $r$ is the number of weeks and $k$ is the number of top topics per week (generally $k = 5$), then each type contains $rk$ slots.
The first $k$ slots correspond to the topics in the first week, the next $k$ slots correspond to the second week, and so on.
To sample a topic set sequence from a single type, we sample topics for each slot independently and take the topic set for each week to be the set of unique topics among that week's slots.
Note that it is possible for this sampling procedure to produce topic sets that are smaller than $k$. 
As in the real Topic API algorithm, when measuring the re-identification risk, any user topic sets that are smaller than size $k$ are padded with random topics.
Finally, our modeled distribution over topic set sequences is a uniform mixture of many types.
To sample a topic set sequence from the mixture, we pick a type uniformly at random and sample a topic set sequence from that type.
\Cref{fig:typesAndSlots} depicts a schematic representation of the model.

\begin{figure}
    \centering
    \includegraphics[width=0.7\columnwidth]{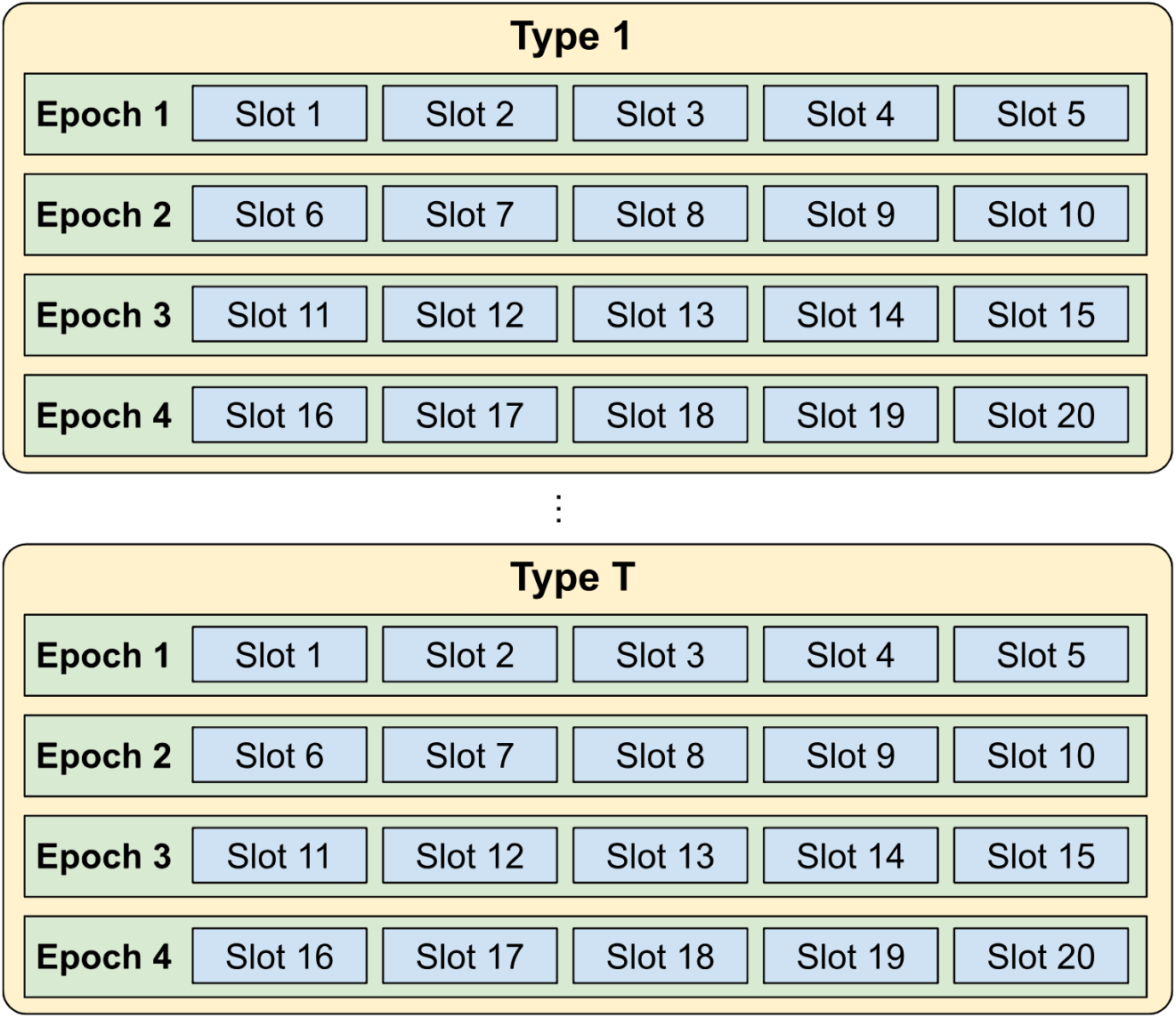}
    \caption{The modeled topic trace distribution is a uniform mixture of $T$ types.
    The topic set for week $e$ in type $t$ is drawn from a collection of $k$ slots.
    Each slot is a distribution over a single topic encoded by a logit vector in $\reals^{|\topics|}$.
    To sample a topic trace, pick a type uniformly at random, sample topics for each of its slots, and then compute the set of unique topics in each week's slots.}
    \label{fig:typesAndSlots}
\end{figure}

Formally, the distribution is parameterized by $\theta \in \reals^{T \times r \times k \times |\topics|}$, where $T$ is the number of types, $r$ is the number of weeks, $k$ is the number of topics per week, and $|\topics|$ is the number of topics in the topic taxonomy.
The vector $\theta[t,i,s,:] = \langle \theta_{t,i,s,1}, \ldots, \theta_{t,i,s,|\topics|} \rangle$ contains the logits for slot $s$ of week $i$ in type $t$.
For simplicity, we assume that the topic ids are the consecutive integers $1$ through $|\topics|$.
Pseudocode for sampling a topic set sequence is given in \Cref{alg:sampling}.

\begin{algorithm}
\noindent\textbf{Input:} Distribution parameters $\theta \in \reals^{T\times r \times k \times |\topics|}$.
\begin{enumerate}
\item Choose type $I$ uniformly at random from $[T]$.
\item Let $p_{i,s} = \operatorname{softmax}(\theta[I,i,s,:])$ for each $i \in [r], s \in [k]$ be the topic probability vector for slot $s$ in week $i$.
\item Sample a topic $x_{i,s} \in \topics$ from the probability vector $p_{i,s}$ for each $i \in [r], s \in [k]$.
\item Let $S_i = \{x_{i,1}, \ldots, x_{i, k}\}$ be the set of unique topics sampled from the slots for each week $i \in [r]$.
\item Return the topic set sequence $(S_i)_{i \in [r]}$.
\end{enumerate}
\caption{Topic Set Sequence Sampling}
\label{alg:sampling}
\end{algorithm}

The model described above has a number of desirable properties.
First, when the number of types $T$ is sufficiently large, this model can fit any statistics computed from a sample of user topic traces.
This is because we can create one type for each user that produces that user's topic trace with probability one.
Then the distribution encoded by the model is exactly the empirical distribution of the sample and statistics computed from the model perfectly match those of the data sample.
However, in most cases we do not need one type per user, and each type should be thought of as a distribution over topics that are typically observed together.
For example, even though we use a dataset containing approximately 100 million users, we are able to achieve a very good fit using $T=500$ types. Second, even though the model is expressive enough to fit the statistics we collect, it is not obvious how to find the optimal parameters.
As we will see in \Cref{sec:optimizing}, each of the statistics from \Cref{sec:stat} can be written as a differentiable function of the model's parameters $\theta$.
We can therefore deploy the extensive machinery of gradient-based optimization to search for model parameters that closely agree with those statistics.

\begin{figure*}[ht!]
    \subcaptionbox{Absolute error percent \label{fig:abs-err}}{
    \includegraphics[width=0.48\textwidth]{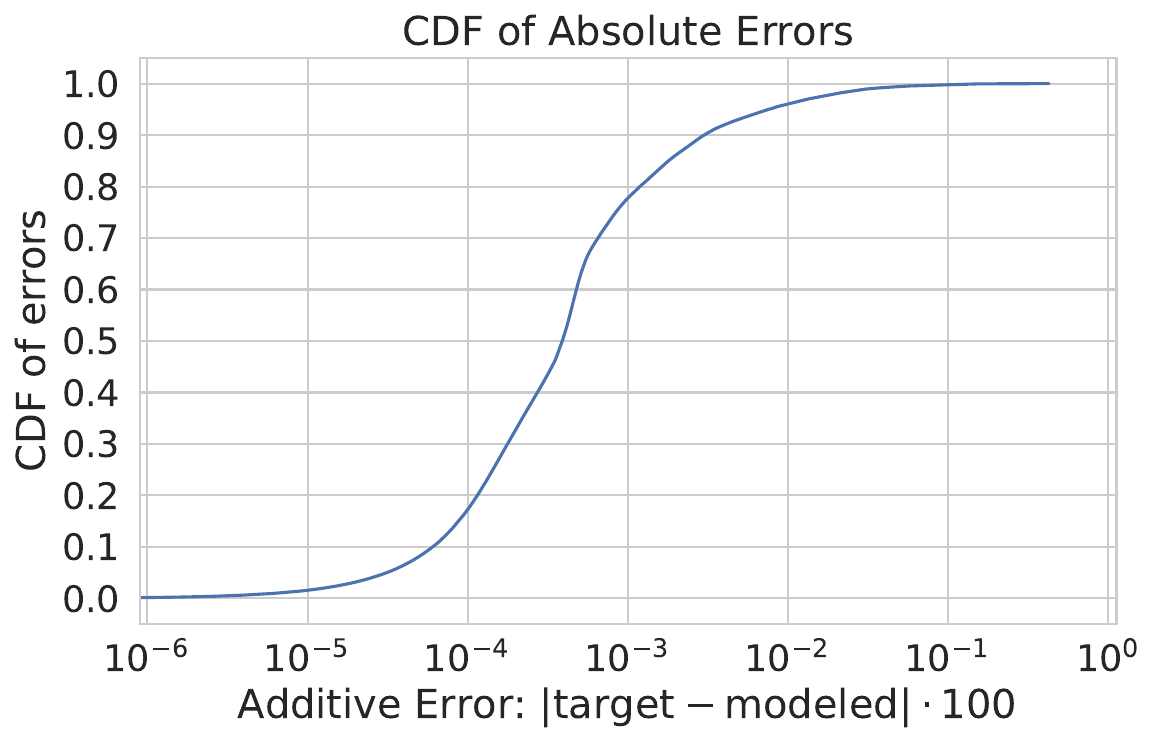}}
    \hfill
    \subcaptionbox{Relative error\label{fig:abs-rel}}{\includegraphics[width=0.48\textwidth]{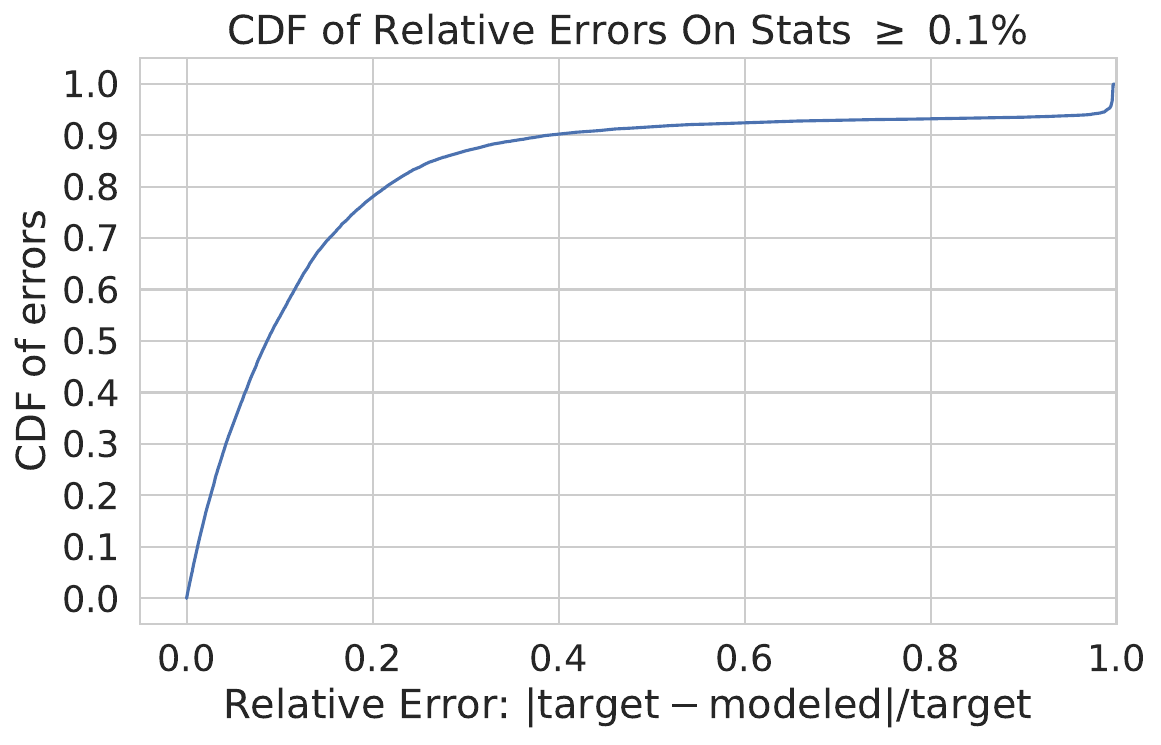}
    }
    \caption{
    Cumulative distribution of absolute and relative errors over the terms of $J(\theta)$.
    Each plot shows the fraction of the terms of $J$ with absolute or relative error below each threshold.
    The relative error plot is restricted to target statistics whose value is at least 0.01\%. 
    }
    \label{fig:absolteAndRelativeErrors} 
\end{figure*}

\begin{figure*}
    \centering
    \subcaptionbox{Signed Error\label{fig:error-pdf-absolute}}{
   \includegraphics[width=0.48\textwidth]{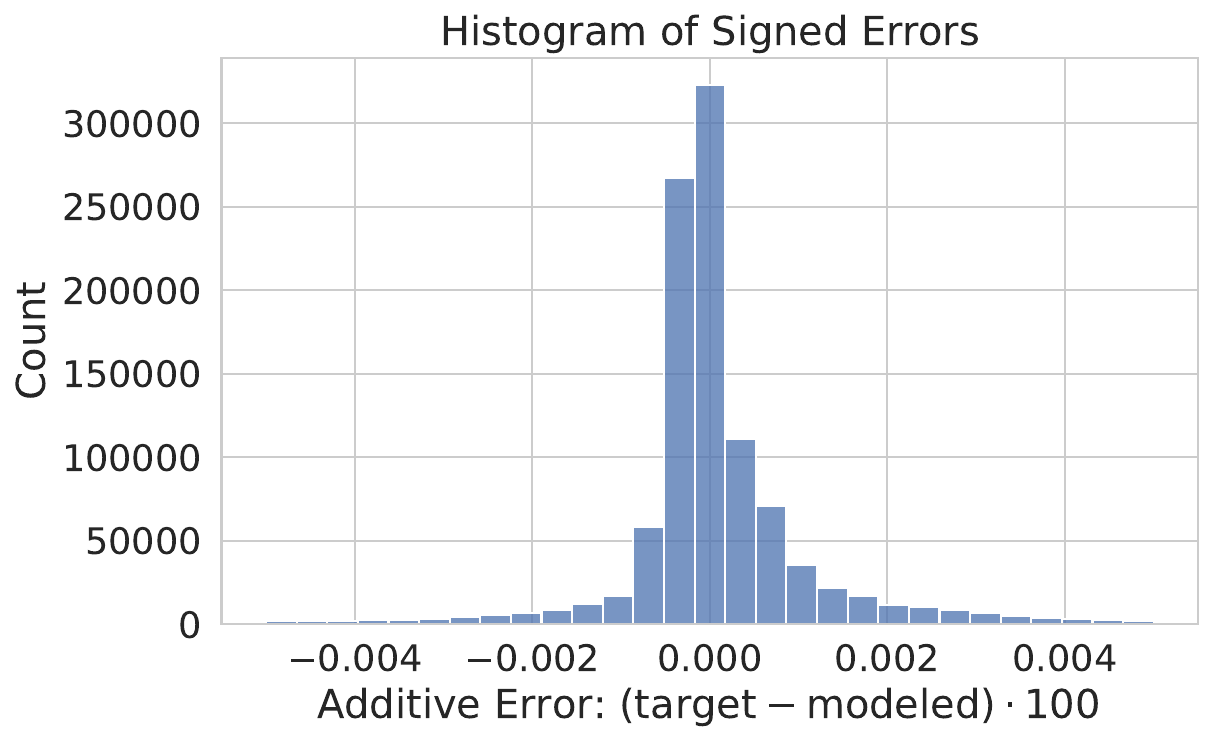}
    }
    \hfill
    \subcaptionbox{Signed Relative Error\label{fig:error-pdf-relative}}{
    \includegraphics[width=0.48\textwidth]{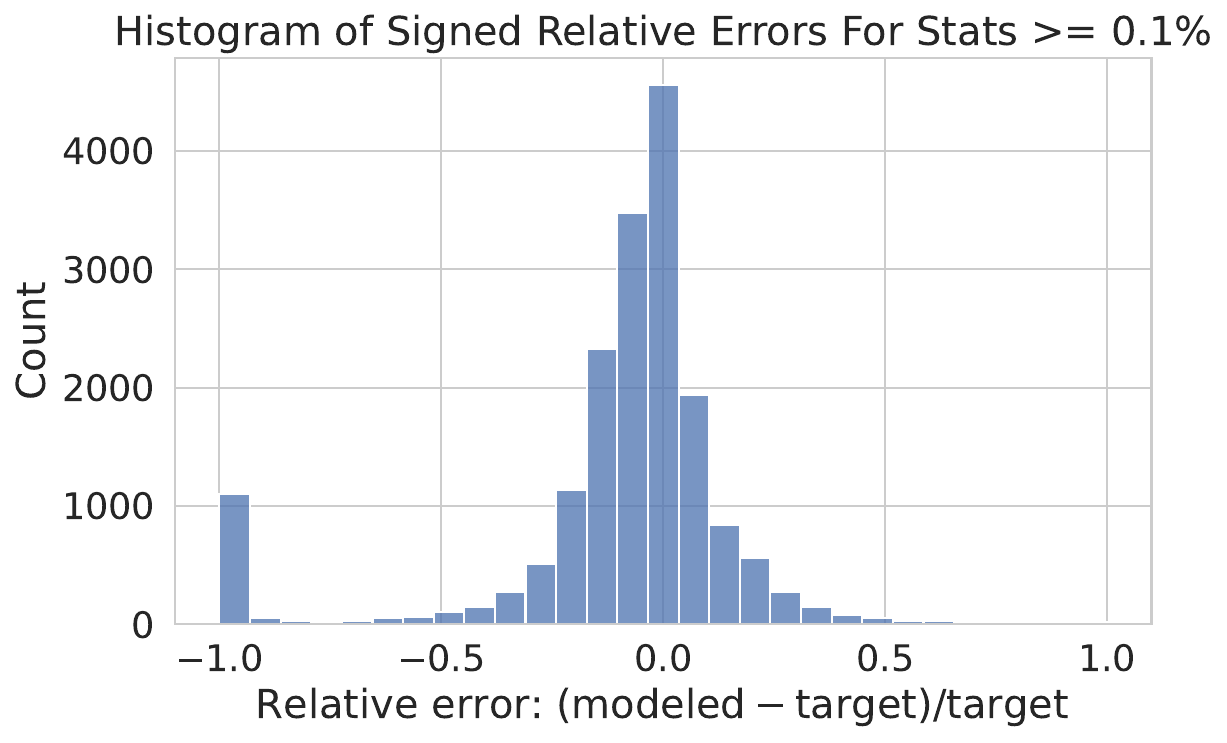}
    }
        \caption{Histograms of the absolute and relative errors over the terms of $J(\theta)$. Notice that the x-axis in Figure~\ref{fig:error-pdf-absolute} is clipped between $[-0.005,0.005]$ to increase clarity, see Figure~\ref{fig:abs-err} for the entire distribution.} 
    \label{fig:errorHistograms}
\end{figure*}

\subsection{Optimizing Model Parameters}
\label{sec:optimizing}

In this section we describe our approach to fitting the parameters of the topic set sequence model described in \Cref{sec:model}.
We define a differentiable objective function of the model parameters such that achieving an objective value of 0 corresponds to perfectly matching the DP statistics computed in \Cref{sec:stat}.
Then, this objective function is optimized using Adam~\cite{kingma2015adam}, which is an adaptive learning rate optimization algorithm originally designed for training deep neural networks.

\begin{restatable}{lemma}{lemDifferentiable}\label{lem:differentiable}
    Let $(S_i)_{i \in [r]}$ be a topic trace sampled from the model with parameters $\theta$.
    There exist functions $\qsingle$, $\qwithin$, and $\qacross$ that are differentiable with respect to $\theta \in \reals^{T \times r \times k \times |\topics|}$ such that
    \begin{enumerate}
        \item For every week $i \in [r]$ and topic $o \in \topics$, we have 
        \[
            \prob(o \in S_i) = \qsingle(\theta; i, o).
        \]
        \item For every week $i \in [r]$ and distinct topics $o_1, o_2 \in \topics$, we have
        \[
        \prob(\{o_1, o_2\} \subset S_i) = \qwithin(\theta; i, o_1, o_2).
        \]
        \item For distinct weeks $i_1, i_2 \in [r]$ and any topics $o_1, o_2 \in \topics$, we have
        \[
        \prob(o_1 \in S_{i_1} \wedge o_2 \in S_{i_2}) = \qacross(\theta; i_1, i_2, o_1, o_2).
        \]
    \end{enumerate}
\end{restatable}
The proof of \Cref{lem:differentiable} essentially uses the inclusion-exclusion principle to write each of the statistics as a polynomial function of $\theta$. The detailed proof is given in \Cref{app:omitted-proofs}.

With this, we are ready to define the optimization problem we solve in order to fit the parameters $\theta$ to the collection of differentially private statistics obtained in \Cref{sec:stat}:
\begin{align*}
    J(\theta) 
    &= 
    \frac{1}{N}
    \sum_{i=1}^r
    \sum_{o=1}^{|\topics|}
    \bigl(
        \qsingle(\theta; i, o) - \qsingle^*(o)
    \bigr)^2 \\
    &+
    \frac{1}{N}
    \sum_{i=1}^r
    \sum_{o_1 = 1}^{|\topics|}
    \sum_{o_2 = o_1+1}^{|\topics|}
    \bigl(
        \qwithin(\theta; i, o_1, o_2) - \qwithin^*(o_1, o_2)
    \bigr)^2 \\
    &+
    \frac{1}{N}    
    \sum_{i=1}^{r-1}
    \sum_{o_1=1}^{|\topics|}
    \sum_{o_2=1}^{|\topics|}
    \bigl(
        \qacross(\theta; i, i+1, o_1, o_2) - \qacross^*(o_1, o_2)
    \bigr)^2,
\end{align*}
where $N = r|\topics| + r{|\topics| \choose 2} + (r-1)|\topics|^2$ is the total number of terms in the objective.

To optimize $J(\theta)$, we perform mini-batch gradient descent where each minibatch corresponds to a subset of the terms of $J(\theta)$.
On each minibatch, we compute the gradient of the selected terms with respect to the model parameters and use Adam to update the parameters.
In our experiments, we use the following optimization hyperparameters: we use 500 types, batches of 8192 terms, an Adam learning rate of 1.0.
We initialize all parameters from a Gaussian distribution with mean zero and standard deviation 0.001 and use 8000 training epochs / data passes. 
The results are stable to changes in the parameters.

\section{Empirical analysis}\label{sec:exp}
\label{sec:dataValidation}
In this section, we evaluate empirically the quality of the synthetic data generated by our model and released publicly~\cite{datarelease}.

\begin{figure}
    \centering
    \includegraphics[width=0.7\textwidth]{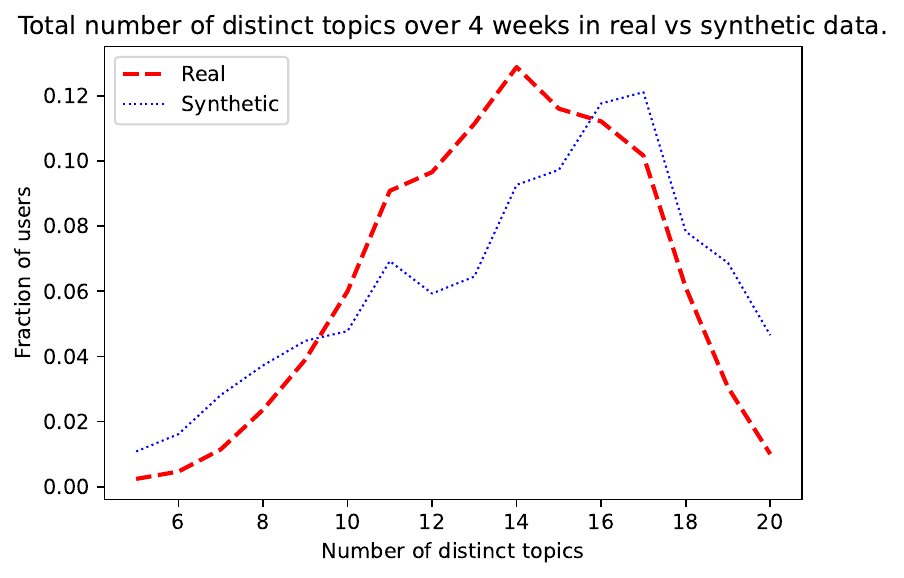}
    \caption{Fraction of users with a given number of distinct topics in their sets over 4 weeks. Notice the strong correlation between the distribution of the real and synthetic data, despite the fact that this statistic is not an input to the model.}
    \label{fig:distinct-topics}
\end{figure}

\paragraph{Data validation}
First, we verify that the model obtained has an output distribution that closely matches the statistics of interest that we have defined in the previous sections. In~\Cref{fig:absolteAndRelativeErrors}, we show the cumulative distribution of the absolute and relative errors of the statistics for a training run of the model. \Cref{fig:abs-err} represents the absolute error in percent points while \Cref{fig:abs-rel} represents the relative error between the statistics computed on the model and on the real data.  
Notice that virtually all statistics have an absolute error lower than $0.1\%$.
We report the relative error in \Cref{fig:abs-rel} for all statistics above $0.1\%$. Notice that the relative error is $\le 20
\%$ for $75\%$ of statistics. This confirms the convergence of our model to an output distribution close to the desired one. 
To gain a better understanding of the error distribution, in~\Cref{fig:errorHistograms}, we present the probability density function. \Cref{fig:error-pdf-absolute} reports the PDF of the absolute error in percent, while~\Cref{fig:error-pdf-relative} reports the PDF of the relative error for statistics with $\ge 0.1\%$ value. Notice how both the absolute and relative error have a large mass centered around $0$ showing that the distribution mean matches the real data.

\paragraph{Additional validation of non-constrained statistics}
So far we tested that our model is able to capture reasonably well the statistics imposed as constraints. In~\Cref{fig:distinct-topics}, instead, we report the distribution of a statistic not explicitly constrained by the model---the number of distinct topics in a user's sets over 4 weeks---on real and synthetic data. Notice the strong correlation between the distribution of the real and synthetic data ($89\%$ Pearson correlation coefficient). This is particularly interesting because this statistic is not part of the inputs of the model, this shows that our model is able to capture reasonably well additional properties of the data beyond the ones explicitly imposed in the training process.

\subsection{Re-identification analysis}
\label{sec:re-id-risk}
In this section, we present the results for the re-identification risk analysis on our synthetic data and show how it represents well the re-identification risk measured on real data. Before presenting the empirical results we define the approach used for defining and measuring re-identification risk.

\subsubsection{Methodology}
In order to show the replicability of prior work using our synthetic data, we follow the same set up of Carey et al.~\cite{carey2023measuring}. We sketch briefly the random-user model for re-identification risk introduced by Carey et al.~\cite{carey2023measuring} and refer to the paper for more details. 
In this model, an attacker observes the Topics API outputs for a set $\users$ of users on a website $w_1$ for a period of $r$ weeks. The adversary has hence a table $$T=\{ (o_1(u, w_1), \ldots o_r(u, w_1)) | u \in \users\},$$ containing the topics of the users for $r$ weeks on the site. Then, a random user $U^*$ (unknown to the attacker) is sampled uniformly from $\users$ and the Topics API output trace $o = (o_1(U^*, w_2), \ldots o_r(U^*, w_2))$ for that user, observed on a different website $w_2$ during the $r$ weeks, is revealed to the adversary. The adversary has to output a prediction $I(T,o) \in \users$ for the identity of the user $u\in \users$ that produced the trace $o$. The re-identification risk is measured as the probability---over the random sampled user $U^*$, randomness of the API and of the attacker---that the  prediction of the attacker is correct, i.e. $$\Pr_{U^*,T,o} (I(T,o)=U^*).$$

The attacker function can be an arbitrary method that attempts to infer the identity of the user from the data. In this work, for replicability, we use two attacks from~\cite{carey2023measuring} that are available\footnote{\url{https://github.com/google-research/google-research/tree/master/re_identification_risk}} as open source code. We now present them.

\paragraph{Hamming Attack}
Given the sampled trace $$o = (o_1(U^*, w_2), \ldots,o_r((U^*, w_2))$$ and the table $$T=\{ (o_1(u, w_1), \ldots o_r(u, w_1)) | u \in \users\},$$ the simple Hamming Attack computes the Hamming distance between the sequence $o$ and any sequence $(o_1(u,w_1),$ $\ldots, o_1(u,w_1))$ from $T$. The Hamming distance between the sequence of $u$ in $w_1$ and sequence $o$ is defined as the number of weeks $i$ such that $o_i(u,w_1) \neq o_i(U^*, w_2)$ (i.e., the number of weeks where output topics are different in the two sites). The simple (unweighted) Hamming Attack predicts the identity of the user generating $o$ as the user in $T$ with the trace with the smallest Hamming distance to $o$ (ties broken randomly).

\subsubsection{Asymmetric Hamming Attack}
The Asymmetric Hamming Attack is a more advanced variant of the (unweighted) Hamming Attack and is based on modeling the joint probability of the outputs on two sites~\cite{carey2023measuring}. This attack predicts as the user producing $o$, the user in $T$ with the trace with the smallest (weighted) Hamming distance to $o$ where the weighted distance employs asymmetric weights $w_{o_1, o_2}$ for $o_1,o_2\in \topics$ that are learned from the data to maximize the probability of correctly identifying the user. We refer to~\cite{carey2023measuring} for more details.

\begin{figure}[t!]
    \centering
    \includegraphics[width=.70\columnwidth]{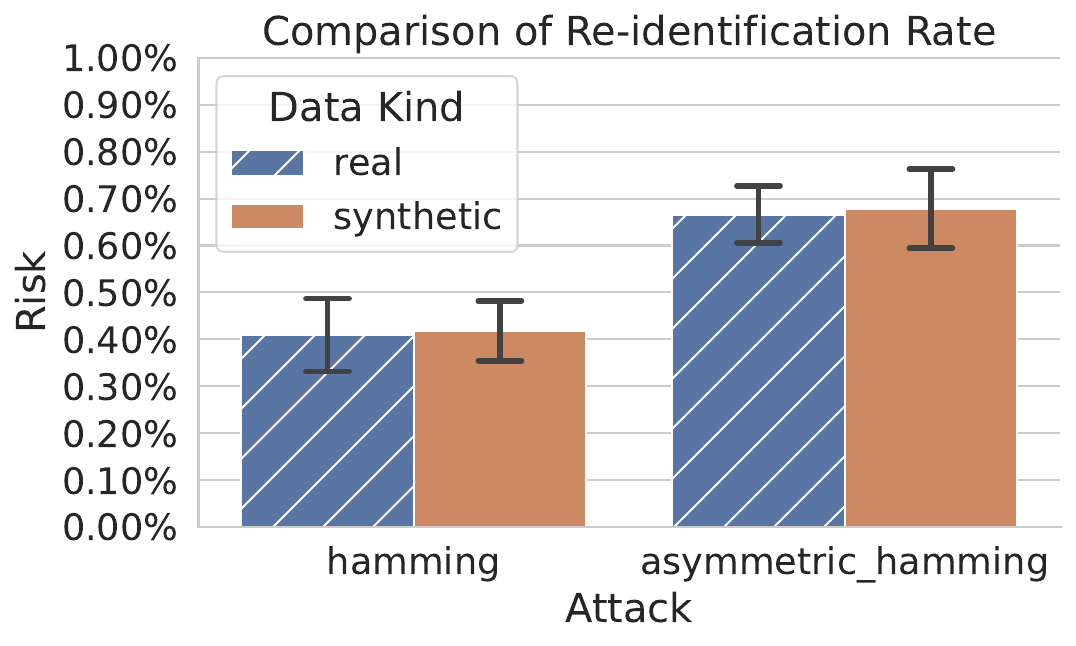}
    \caption{The re-identification rates of the Hamming and Asymmetric Hamming attacks on the real and synthetic data after 4 weeks of observation.
    Error bars show $\pm$ one standard deviation over 10 independent runs.}
    \label{fig:reidentificationRates}
\end{figure}

\subsubsection{Results}
We now present the results of the re-identification risk experiment. To compare our synthetic data with the real world data, we obtain a sample of $10$ million users' traces observed over $4$ weeks from the real user data, and compare them with $10$ million synthetic traces (of the same number of weeks) generated by our synthetic data model. On each of these two datasets (separately) we repeat the $2$ attacks presented before on a sample of $10240$ random users $U^*$ and define as re-identification risk the fraction of correct answers for the attack / dataset. This process is further repeated for $10$ independent trials (for the synthetic data, we repeat independently the training of the model on the DP statistics and the sampling of the data from the model as well). 

Our results are presented in \Cref{fig:reidentificationRates} which compares the  re-identification risk of the two attacks (Hamming Attack and Asymmetric Hamming Attack) on the real user data and on our synthetic dataset. The bars show the mean re-id risk. The error bars, instead, represent $1$ standard deviation of the re-identification risk (over the $10$ repetitions of the experiment).

We make a series of observations. First, we observe that we closely replicate the observations of prior work~\cite{carey2023measuring} concerning the re-identification risk of the two attacks on real data. The Asymmetric Attack is more sophisticated and has a higher re-identification risk. However, both attacks have a re-id risk less than $1\%$.

More interesting, it is possible to observe that re-identification risk measured by each attack on the synthetic data is remarkably close to the result of the {\it same} attack on real data. In both attacks, the re-identification risk on the synthetic data is within 1 standard deviation from the mean re-identification of the same attack on real data. This highlights that our model, despite its necessary simplifications, is sufficiently accurate to characterizing the re-identification risk of the Topics API under the  attacks studied.

\section{Conclusion, limitations of our work and future work}
In this paper we presented a novel methodology for generating  synthetic data enabling the study of the privacy properties of the Topics API. Using our method we release an open source dataset~\cite{datarelease} that is derived from real user data while providing strong privacy protection. As shown in our experiments, this dataset is accurate enough to preserve important properties of the data including allowing us to replicate a prior re-identification analysis performed on a proprietary datasets. 

Notwithstanding these positive results, it is important to highlight that like all synthetic data releases, our data can not replicate exactly {\it all} properties of the original (private) data. In fact, while we show that many important statistics are accurately preserved, it is indeed possible that the model might not preserve other statistics (including important statistics that are not considered in our work, but are relevant to other studies). In our model, we made several assumptions, especially on the time-stationarity of the user data. We stress that the stationarity assumption in particular is crucial for the validity of our results. While we partially validated these assumptions on the data available, it is possible that the data may evolve over longer time scales (or in the future), diverging from our measured statistics. Moreover, as we observed, since preserving an arbitrary distribution is infeasible due to its extremely large parameter space, our model must necessarily have neglected some aspects of the real dataset. While our methodology focused on replicating some prior attacks on re-identification risk, it is possible that the synthetic data we released may be unable to correctly replicate other methodologies.

While this synthetic data release has its inherent limitations, we believe that our work will benefit the study of the privacy properties of the Topics API providing researchers with access to large-scale and sufficiently realistic data. As future work, we believe that our methodology could be generalized to several Privacy Preserving Ads APIs providing insights in other efforts in this space. We also believe that future work on improving our synthetic data methodology, by for instance, reducing the assumptions needed for its validity, is of great interest.

\section*{Acknowledgment}
Adel Javanmard is supported in part by the NSF Award
DMS-2311024, the Sloan fellowship in Mathematics, an
Adobe Faculty Research Award, an Amazon Faculty Research Award, and an iORB grant from USC Marshall
School of Business. The authors are grateful to anonymous reviewers for their
feedback on improving this paper.


\printbibliography

\appendix

\section{Omitted Proofs}
\label{app:omitted-proofs}

\begin{figure}[b]
\centering
    \subcaptionbox{Top 10 Within Week Pairs \label{fig:top-within}}{
    \begin{tabular}{lr}
    \toprule
    Name & Rate \\
    \midrule
    '/Arts \& Entertainment' and '/News & 11.77\% \\
    '/News' and '/Sports & 8.58\% \\
    '/Arts \& Entertainment' and '/Internet \& Telecom & 6.74\% \\
    '/Internet \& Telecom' and '/News & 5.27\% \\
    '/Arts \& Entertainment' and '/Games/Computer \& Video Games & 5.00\% \\
    '/Sports' and '/Sports/Soccer & 4.60\% \\
    '/Arts \& Entertainment' and '/Arts \& Entertainment/TV \& Video & 4.58\% \\
    '/News' and '/Shopping & 4.41\% \\
    '/News' and '/Arts \& Entertainment/TV \& Video & 4.18\% \\
    '/News' and '/News/Local News & 4.17\% \\
    \bottomrule
    \end{tabular}
    \vspace{.4in}
    }
    
    \subcaptionbox{Top 10 Across Weeks Pairs \label{fig:top-across}}{
    \begin{tabular}{lr}
    \toprule
    Name & Rate \\
    \midrule
    '/News' then '/News & 22.01\% \\
    '/Arts \& Entertainment' then '/Arts \& Entertainment & 12.27\% \\
    '/Shopping' then '/Shopping & 12.05\% \\
    '/Sports' then '/Sports & 11.66\% \\
    '/Games/Computer \& Video Games' then '/Games/Computer \& Video Games & 10.84\% \\
    '/News' then '/Arts \& Entertainment & 10.16\% \\
    '/Arts \& Entertainment' then '/News & 10.11\% \\
    '/Sports' then '/News & 8.20\% \\
    '/News' then '/Sports & 7.97\% \\
    '/News' then '/Shopping & 5.86\% \\
    \bottomrule
    \end{tabular}
    }
    \caption{Topics and pairs of topics with the largest values of  $\qwithin^*$, and $\qacross^*$. \label{fig:stats-app} \vspace{1.5cm}}
\end{figure}

\lemDifferentiable*
\begin{proof}
Let $\theta \in \reals^{T \times r \times k \times |\topics|}$ be the model parameters and let $I, S_1, \ldots, S_r$ be the type index and sequence of topic sets sampled by \Cref{alg:sampling} using parameters $\theta$.
Define $P_\theta \in [0,1]^{T \times r \times k \times |\topics|}$ by $P_\theta[t,i,s,:] = \operatorname{softmax}(\theta[t, i, s, :])$ so that $P_\theta[t, i, s, o]$ is the probability of sampling topic $o$ in slot $s$ of week $i$ for type $t$.

We begin by calculating the probability of \emph{not} observing any topic from an arbitrary set $A \subset \topics$ in week $i$ conditioned on the event $I = t$.
This quantity will appear several times in $\qsingle$, $\qwithin$, and $\qacross$.
Recall that $x_{i,s}$ is the topic sampled by \Cref{alg:sampling} in slot $s$ for week $i$ using the slot distributions for type $I$.
For any type $t \in [T]$, week $i \in [r]$, and set of topics $A \subset \topics$, we have
\begin{align*}
    \prob(S_i \cap A = \emptyset \mid I = t)
    &= \prob(\forall s \in [k] \,:\, x_{i,s} \not \in A \mid I = t) \\
    &= \prod_{s=1}^k \prob(x_{i,s} \not \in A \mid I = t) \\
    &= \prod_{s=1}^k \left(
        1 - \sum_{x \in A} P[t,i,s,x]
    \right),
\end{align*}
where the second equality follows from the fact that the topics sampled in the slots for week $i$ are independent conditioned on the type.
Let $N(\theta; t, i, A) = \prod_{s=1}^k(1 - \sum_{x \in A}P[t,i,s,x]) = \prob(S_i \cap A = \emptyset \mid I = t)$.
Since $N$ is a differentiable function of $P$, and $P$ is obtained by applying $\operatorname{softmax}$, which is also differentiable, to $\theta$, it follows that $N$ is a differentiable function of $\theta$.
With this, we turn to calculating expressions for each statistic computed in \Cref{sec:stat}.

Now consider the probability of observing topic $o \in \topics$ in the topic set for week $i \in [r]$.
We have
\begin{align*}
    \prob(o \in S_i)
    &= \sum_{t=1}^T \prob(o \in S_i \mid I = t) \prob(I = t) \\
    &= \frac{1}{T} \sum_{t=1}^T \bigl(1 - \prob(o \not \in S_i \mid I = t)\bigr)\\
    &= \frac{1}{T} \sum_{t=1}^T \bigl(1 - N(\theta; t, i, \{o\})\bigr) \\
    &=: \qsingle(\theta; i, o).
\end{align*}

Calculating the probability that two distinct topics $o_1, o_2 \in \topics$ both appear in week $i \in [r]$ can be done using the inclusion-exclusion principle:
\begin{align*}
    &\prob(\{o_1, o_2\}\subset S_i)
    = \frac{1}{T} \sum_{t=1}^T \prob(\{o_1, o_2\} \subset S_i \mid I = t) \\
    &= \frac{1}{T} \sum_{t=1}^T
    \bigl(1 - \prob(o_1 \not \in S_i \vee o_2 \not \in S_i \mid I = t)\bigr) \\
    &= \frac{1}{T}\sum_{t=1}^T\bigl(
    1 - \prob(o_1 \not \in S_i)
     - \prob(o_2 \not \in S_i)
     + \prob(o_1 \not \in S_i \wedge o_2 \not \in S_i)
    \bigr) \\
    &= \frac{1}{T} \sum_{t=1}^T \bigl(
        1
        - N(\theta; t, i, \{o_1\})
        - N(\theta; t, i, \{o_2\})
        + N(\theta; t, i, \{o_1, o_2\})
    \bigr)\\
    &=:\qwithin(\theta; i, o_1, o_2).
\end{align*}

\begin{figure*}[]
\centering
    \subcaptionbox{Single Topic Frequencies}{
    \begin{tabular}{l|cccc}
    \toprule
     & week 1 & week 2 & week 3 & week 4 \\
    \midrule
    week 1 & 1.000000 & 0.999865 & 0.998922 & 0.998571 \\
    week 2 & & 1.000000 & 0.999311 & 0.998756 \\
    week 3 & & & 1.000000 & 0.999546 \\
    week 4 & & & & 1.000000 \\
    \bottomrule
    \end{tabular}
    \vspace{.4in}
    }
    
    \subcaptionbox{Topic Pair Frequencies}{
    \begin{tabular}{l|cccc}
    \toprule
     & week 1 & week 2 & week 3 & week 4 \\
    \midrule
    week 1 & 1.000000 & 0.999682 & 0.998101 & 0.997261 \\
    week 2 &  & 1.000000 & 0.998816 & 0.997761 \\
    week 3 &  &  & 1.000000 & 0.999091 \\
    week 4 &  &  &  & 1.000000 \\
    \bottomrule
    \end{tabular}
    \vspace{.4in}
    }
    
    \subcaptionbox{Topic Transitions}{
    \begin{tabular}{l|ccc}
    \toprule
     & week 1 to 2 & week 2 to 3 & week 3 to 4 \\
    \midrule
    week 1 to 2 & 1.000000 & 0.999384 & 0.998256 \\
    week 2 to 3 &  & 1.000000 & 0.999186 \\
    week 3 to 4 &  &  & 1.000000 \\
    \bottomrule
    \end{tabular}
    }
    \caption{Pearson correlation coefficient between statistics computed in two periods of time. Pearson correlation coefficient is between $-1$ perfect negative correlation, $+1$ perfect positive correlation. $0$ means no correlation.}
    \label{fig:stationarity}
\end{figure*}

The probability of observing topic $o_1 \in \topics$ in week $i_1 \in [r]$ and topic $o_2 \in \topics$ in week $i_2 \in [r]$ when $i_1 \neq i_2$ can be calculated using the fact that, conditioned on the type $I$, these events are independent.
\begin{align*}
    \prob(o_1 \in S_{i_1} \wedge& o_2 \in S_{i_2})
    = \frac{1}{T} \sum_{t=1}^T \prob(o_1 \in S_{i_1} \wedge o_2 \in S_{i_2} \mid I = t) \\
    &= \frac{1}{T} \sum_{t=1}^T \prob(o_1 \in S_{i_1} \mid I = t)
    \cdot
    \prob(o_2 \in S_{i_2} \mid I = t) \\
    &= \frac{1}{T} \sum_{t=1}^T
    \bigl(1 - N(\theta; t, i_1, \{o_1\})\bigr)
    \cdot
    \bigl(1 - N(\theta; t, i_2, \{o_2\})\bigr)\\
    &=: \qacross(\theta; i_1, i_2, o_1, o_2).
\end{align*}

Finally, all three functions $\qsingle$, $\qwithin$, and $\qacross$ are differentiable functions of $\theta$, since they are differentiable in $N$.
\end{proof}

\lemSensitivity*
\begin{proof}
Let $$D' = D \cup \{(S_0(v), S_1(v))\}$$ for a user $v\notin \users$.
Notice that $f_{11}(D')-f_{11}(D)$ is vector with only 0,1 entries where the $1$ correspond to the ${k \choose 2}$ counts for topics pairs $o_i, o_j \in S_0(v)$  in the corresponding positions in $f_{11}$.

The same is true for  $f_{22}$, while for $f_{12}$ we have a 1 in the $k^2$ counts for topics pairs $o_i \in S_0(v)$ and $o_j \in S_1(v)$ ( in the corresponding positions in $f_{12}$). 
Hence the $\ell_2$ norm of the difference is $\sqrt{{k \choose 2}}$ for $f_{11}$ and $f_{22}$ and $k$ for $f_{12}$.
\end{proof}

\section{Additional plots}
\label{app:additional-plots}

\paragraph{DP statistics}
In Figure~\ref{fig:stats-app} we present additional statistics on the top $10$
pairs of topics within a week (\Cref{fig:top-within}) and top $10$  pairs of topics across two consecutive weeks (\Cref{fig:top-across}). Notice how most top pairs of topics in the $\qwithin^*$ and $\qacross^*$ are from the top topics by frequency~\Cref{fig:stats-top} as expected. Also notice that due to the correlation of topics in two consecutive weeks several top pairs in~\Cref{fig:top-across} have the same topic.

\section{Stationarity of Topics}\label{app:stationarity}
This section studies the extent to which the statistics we measure are stationary over time.
We define the following single-week versions of each statistic:
\begin{align*}
    \qsingle^*(o, j) &= \prob_{U}(o \in S_j(U)) \\
    \qwithin^*(o_1, o_2, j) &= \prob_{U}(\{o_1, o_2\} \subset S_j(U))\\
    \qacross^*(o_1, o_2, j) &= \prob_{U}(o_1 \in S_{j-1}(U) \wedge o_2 \in S_{j}(U)).
\end{align*}
These are identical to the statistics defined in \Cref{sec:stat} except that instead of averaging over weeks of training data, they define the statistics from a specific week, or a specific week transition in the case of $\qacross^*$.
In \Cref{fig:stationarity}, we provide the Pearson correlation coefficient between each pair of weeks (or pair of week transitions) for each of the statistics $\qsingle^*$, $\qwithin^*$, and $\qacross^*$.
We see that there is a very high degree of correlation between the statistics, supporting the claim that these statistics are very close to being stationary over time. These results are reported in~
\Cref{fig:stationarity}. Notice that for any statistic and for any pair of periods the Pearson correlation coefficient is $\ge 99.8\%$ thus supporting the stationarity assumption at the basis of our model.

\end{document}